\documentclass[journal]{IEEEtran}
\usepackage{cite}

\ifCLASSINFOpdf
\else
\fi
\usepackage{tikz}

\usepackage[cmex10]{amsmath}
\usepackage{tikz}
\usepackage{pgfplots}
\pgfplotsset{width=8cm, compat=1.3, every axis/.append style={line width=0.7pt}}

\usepackage{amssymb}
\usepackage{amsthm}

\RequirePackage{mathtools}

\newtheorem{theorem}{Theorem}[section]
\newtheorem{proposition}[theorem]{Proposition}

\renewcommand{\epsilon}{\varepsilon}

\newcommand{\supp}{\textnormal{supp}}
\newcommand{\abs}[1]{\lvert#1\rvert} 
\newcommand{\card}[1]{\abs{#1}} 
\newcommand{\bigcard}[1]{\bigl\lvert#1\bigr\rvert}

\newcommand{\comp}[1]{{#1}^{\textnormal{c}}} 
\newcommand{\setcomp}[1]{{#1}^{\textnormal{c}}} 
\newcommand{\set}[1]{\mathcal{#1}} 
\newcommand{\eqdef}{\triangleq} 
\newcommand{\E}[1]{\operatorname{E}[#1]} 
\newcommand{\bigE}[1]{\operatorname{E}\bigl[#1\bigr]} 
\newcommand{\naturals}{\mathbb{N}}

\newcommand{\rhotrans}{\tilde{\rho}}

\begin{document}
%
\title{Encoding Tasks and R\'enyi Entropy}
%
%
%

\author{Christoph~Bunte and Amos~Lapidoth
\thanks{C. Bunte and A. Lapidoth are with the Signal and Information Processing
Laboratory at ETH Zurich. E-mail: \{bunte,lapidoth@isi.ee.ethz.ch\}.}
}

\maketitle

\begin{abstract}
A task is randomly drawn from a finite set of tasks and is described using a fixed
number of bits. All the tasks that share its description must be performed. Upper and
lower bounds on the minimum $\rho$-th moment of the number of performed tasks 
are derived. The case where a sequence of tasks is produced by a source and~$n$ 
tasks are jointly described using~$nR$ bits is considered. If $R$ is larger than the R\'enyi entropy rate of the source of order
$1/(1+\rho)$ (provided it exists), 
then the $\rho$-th moment of the ratio of performed tasks to $n$
can be driven to one as $n$ tends to infinity. If $R$ is smaller than the
R\'enyi entropy rate, this moment tends to infinity. 
The results are generalized to account for the presence of side-information.
In this more general setting, the key quantity is a conditional version of
R\'enyi entropy that was introduced by Arimoto.
For IID sources two additional extensions are solved, one
of a rate-distortion flavor and the other where different tasks may have
different nonnegative costs. 
Finally, a divergence that was identified by Sundaresan as a mismatch
penalty in the Massey-Arikan guessing problem is shown to play a 
similar role here.
\end{abstract}

\begin{IEEEkeywords}
Divergence, R\'enyi entropy, R\'enyi entropy rate, mismatch, source coding,
tasks.
\end{IEEEkeywords}

%
\IEEEpeerreviewmaketitle

\section{Introduction}
\label{sec:introduction}
A task $X$ that is drawn from a finite set of tasks $\set{X}$ according to some
probability mass function (PMF) $P$ is to be described using a
fixed number of bits. The least number of bits needed for an unambiguous description
is the base-2 logarithm of the total number of tasks in $\set{X}$
(rounded up to the nearest integer). 
When fewer bits are available, the classical
source coding approach is to provide descriptions for the tasks with
the largest or with the ``typical'' probabilities only. This has the obvious drawback that less common,
or atypical, tasks will never be completed. 
For example, if $\set{X}$ comprises all possible household chores, then ``wash
the dishes'' will almost certainly occur more frequently than ``take out the
garbage'', but most people would agree that the latter should not be neglected.

The classical approach is not so well-suited here because it does not take into
account the fact that not performing an unlikely but critical task may have
grave consequences, and that performing a superfluous task often causes little or no
harm. A more natural approach in this context is to partition the set of tasks
into subsets. If a particular task needs to be completed, then the subset containing
it is described and all the tasks in this subset are performed. 
This approach has the disadvantage that tasks are sometimes completed superfluously, but
it guarantees that critical tasks, no matter how atypical, are never neglected 
(provided that the number of subsets in the partition of $\set{X}$ does not exceed $M$, when
$\log M$ is the number of bits available to describe them). 
One way to partition
the set of tasks is to provide distinct descriptions for the typical tasks and
to group together the atypical ones. We will see, however, that this may not
always be optimal. 

If we assume for simplicity that all tasks require an equal amount of effort, then 
it seems reasonable to choose the subsets so as to minimize the expected number of
performed tasks. Ideally, this expectation is close to one. More generally, we
look at the $\rho$-th moment of the number of performed tasks, where
$\rho$ may be any positive number. 
Phrased in mathematical terms, we consider encoders of the form
\begin{equation}
f\colon\set{X}\to\{1,\ldots,M\},
\end{equation}
where $M$ is a given positive integer. Every such encoder gives rise to a
partition of $\set{X}$ into $M$ disjoint subsets
\begin{equation}
f^{-1}(m) = \bigl\{x\in \set{X}: f(x) = m\bigr\},\quad m\in\{1,\ldots, M\}.
\end{equation}
Here $f(x)$ is the description of the task $x$, and the set $f^{-1}(f(x))$
comprises all the tasks sharing the same description as $x$, i.e., the set
of tasks that are performed when $x$ is required.

We seek an $f$ that minimizes the $\rho$-th moment of the cardinality of
$f^{-1}(f(X))$, i.e., 
\begin{equation}
\label{eq:moment_of_tasks}
\bigE{\card{f^{-1}(f(X))}^\rho} = \sum_{x\in\set{X}} P(x)
\card{f^{-1}(f(x))}^\rho.
\end{equation}
This minimum is at least 1 because $X\in f^{-1}(f(X))$; it is nonincreasing in $M$ 
(because fewer tasks share the same description when $M$ grows); and it is equal to
one for~$M \geq \card{\set{X}}$ (because then $\set{X}$ can be partitioned into
singletons).
Our first result is a pair of lower and upper bounds on this minimum. 
The bounds are expressed in terms of the \emph{R\'enyi entropy of $X$ of
order $1/(1+\rho)$}
\begin{equation}
\label{eq:renyi}
H_{\frac{1}{1+\rho}}(X) = \frac{1+\rho}{\rho} \log
\sum_{x\in\set{X}}P(x)^{\frac{1}{1+\rho}}.
\end{equation}
Throughout $\log(\cdot)$ stands for $\log_2(\cdot)$, the logarithm to base $2$. 
For typographic reasons we henceforth use the notation
\begin{equation}
\rhotrans=\frac{1}{1+\rho},\quad \rho>0.
\end{equation}
\begin{theorem}
\label{thm:oneshot}
Let $X$ be a chance variable taking value in a finite set $\set{X}$, and let
$\rho>0$. 
\begin{enumerate}
\item For all positive integers $M$ and every $f\colon \set{X}\to\{1,\ldots, M\}$, 
\begin{equation}
\label{eq:oneshot_lb}
\bigE{\card{f^{-1}(f(X))}^\rho} \geq 2^{\rho(H_{\rhotrans}(X)-\log M)}.
\end{equation}
\item For all integers $M > \log\card{\set{X}}+2$ there exists
$f\colon\set{X}\to\{1,\ldots,M\}$ such that
\begin{equation}
\label{eq:oneshot_ub}
\bigE{\card{f^{-1}(f(X))}^\rho} <1+2^{\rho(H_{\rhotrans}(X)-\log \widetilde{M})},
\end{equation}
where $\widetilde{M}=(M-\log\card{\set{X}}-2)/4$. 
\end{enumerate}
\end{theorem}
A proof is provided in Section~\ref{sec:proof}.
Theorem~\ref{thm:oneshot} is particularly useful when applied to the 
case where a sequence of tasks is produced by a
source~$\{X_i\}_{i=1}^\infty$ with alphabet~$\set{X}$ and the first $n$ tasks
$X^n=(X_1,\ldots,X_n)$ are jointly described using~$nR$ bits (the number $R$ is
the \emph{rate} of the description in bits per task and can be any nonnegative number):
\begin{theorem}
\label{thm:main}
Let $\{X_i\}_{i=1}^\infty$ be a source with finite alphabet~$\set{X}$, and let
$\rho>0$. 
\begin{enumerate}
\item If $R>\limsup_{n\to\infty} H_{\rhotrans}(X^n)/n$, 
then there exist encoders $f_n\colon \set{X}^n \to \{1,\ldots, 
2^{nR}\}$ such that\footnote{Throughout $2^{nR}$ stands for $\lfloor
2^{nR}\rfloor$.}
\begin{equation}
\lim_{n\to\infty}\bigE{\card{f^{-1}_n(f_n(X^n))}^\rho} =1.
\end{equation}
\item If $R<\liminf_{n\to\infty} H_{\rhotrans}(X^n)/n$, 
then for any choice of encoders $f_n \colon \set{X}^n \to \{1,\ldots,2^{nR}\}$, 
\begin{equation}
\lim_{n\to\infty}\bigE{\card{f^{-1}_n(f_n(X^n))}^\rho} =\infty.
\end{equation}
\end{enumerate}
\end{theorem}
\begin{proof}
On account of Theorem~\ref{thm:oneshot}, for all $n$ large enough so that
$2^{nR} > n \log
\card{\set{X}}+2$, 
\begin{multline}
\label{eq:proof_thm_main}
2^{n\rho\bigl(\frac{H_{\rhotrans}(X^n)}{n}- R\bigr)} \leq \min_{f_n\colon \set{X}^n \to
\{1,\ldots,2^{nR}\}} \bigE{\card{f_n^{-1}(f_n(X^n))}^\rho}\\
< 1+
2^{n\rho\bigl(\frac{H_{\rhotrans}(X^n)}{n} - R+\delta_n\bigr)},
\end{multline}
where $\delta_n \to 0$ as $n\to\infty$. 
\end{proof}
When it exists, the limit 
\begin{equation}
H_{\alpha}(\{X_i\}_{i=1}^\infty)\eqdef\lim_{n\to\infty} \frac{H_{\alpha}(X^n)}{n}
\end{equation}
is called the \emph{R\'enyi entropy rate of $\{X_i\}_{i=1}^\infty$ of order $\alpha$}. 
It exists for a large class of sources, including time-invariant Markov
sources~\cite{rached2001renyi,pfister2004renyi}. 

If we assume that every $n$-tuple of tasks in $f^{-1}_n(f_n(X^n))$ is performed
(even if this means that some tasks are performed multiple times) and thus that the
total number of performed tasks is $n$ times $\card{f_n^{-1}(f_n(X^n))}$, then 
Theorem~\ref{thm:main} furnishes the following operational characterization of the R\'enyi entropy
rate for all orders in~$(0,1)$. 
For all rates above the R\'enyi entropy rate of order $1/(1+\rho)$, 
the $\rho$-th moment of the ratio of performed tasks to $n$ can be driven to 
one as $n$ tends to infinity. For all rates below it, this moment
grows to infinity. 
In fact, the proof of Theorem~\ref{thm:main} shows that for
large $n$ it
grows exponentially in~$n$ with exponent approaching
\begin{equation}
\rho\bigl(H_{\rhotrans}(\{X_i\}_{i=1}^\infty) - R\bigr). 
\end{equation}
More precisely,~\eqref{eq:proof_thm_main} shows that for all rates $R< H_{\rhotrans}(\{X_i\}_{i=1}^\infty)$, 
\begin{multline}
\lim_{n\to\infty} \frac{1}{n} \log \min_{f_n\colon \set{X}^n \to
\{1,\ldots,2^{nR}\}} \bigE{\card{f_n^{-1}(f_n(X^n))}^\rho}\\ = 
\rho\bigl(H_{\rhotrans}(\{X_i\}_{i=1}^\infty) - R\bigr). 
\end{multline}
Note that for IID sources the R\'enyi entropy rate reduces to the R\'enyi entropy because
in this case~$H_{\rhotrans}(X^n) = n H_{\rhotrans}(X_1)$.
Other operational characterizations of the R\'enyi entropy rate were given 
in~\cite{rached2001renyi,rached1999renyi,chen2001csiszar,pfister2004renyi,malone2004guesswork,hanawal2011guessing}, and of the R\'enyi entropy 
in~\cite{campbell1965coding, renyi1965foundations, csiszar1995generalized,
arikan1996inequality}. 

The connection between the problem of encoding tasks and 
the Massey-Arikan guessing
problem~\cite{massey1994guessing,arikan1996inequality} is explored 
in~\cite{bracher2014distributed}.

The operational characterization of R\'enyi entropy provided by
Theorem~\ref{thm:main} (applied to IID sources) reveals many of the known properties of
R\'enyi entropy (see, e.g., \cite{csiszar1995generalized,shayevitz2010note}). 
For example, it shows that $H_{\rhotrans}(X)$ is nondecreasing in $\rho$
because~$\xi^\rho$ is nondecreasing in $\rho$ when $\xi\geq 1$. It also shows
that
\begin{equation}
\label{eq:renyi_bounds}
H(X) \leq H_{\rhotrans}(X) \leq \log \card{\supp(P)},
\end{equation}
where $H(X)$ denotes the Shannon entropy and $\supp(P)=\{x:P(x)>0\}$ denotes the support of $P$. 
Indeed, if $R< H(X)$, then, by the converse part of the classical source-coding
theorem~\cite[Theorem~3.1.1]{gallager1968information}
\begin{equation}
\lim_{n\to\infty}\Pr\bigl(\card{f_n^{-1}(f_n(X^n))} \geq 2\bigr) =1, 
\end{equation}
which implies that the $\rho$-th moment of $\card{f_n^{-1}(f_n(X^n))}$ cannot
tend to one as $n$ tends to infinity. 
And if $R\geq \log \card{\supp(P)}$, then every $n$-tuple of tasks that occurs
with positive probability can be given a distinct description so for every
$n$
\begin{equation}
\min_{f_n\colon \set{X}^n \to
\{1,\ldots,2^{nR}\}} \bigE{\card{f_n^{-1}(f_n(X^n))}^\rho}=1.
\end{equation}

The limit
\begin{equation}
\lim_{\rho\to\infty} H_{\rhotrans}(X) = \log \card{\supp(P)}
\end{equation}
follows from our operational characterization of R\'enyi entropy as follows. 
If $R< \log \card{\supp(P)}$, then, by the pigeonhole-principle, for any choice
of $f_n \colon \set{X}^n \to \{1,\ldots, 2^{nR}\}$ there must exist
some $x^n_0 \in \supp(P^n)$ for which 
\begin{equation}
\card{f_n^{-1}(f_n(x^n_0))}\geq 2^{n(\log\card{\supp(P)}-R)}. 
\end{equation}
Since $P^n(x^n_0) \geq p_{\text{min}}^n$, where
$p_{\text{min}}$ denotes the smallest nonzero probability of any source symbol,
we have
\begin{align}
\bigE{\card{f_n^{-1}(f_n(X^n))}^\rho}&\geq 
P^n(x^n_0)\card{f_n^{-1}(f_n(x^n_0))}^\rho\\
&\geq 2^{n\rho (\log\card{\supp(P)}-R+\rho^{-1} \log p_{\text{min}})}.
\end{align}
For all sufficiently large $\rho$ the RHS tends to infinity as $n\to\infty$, which proves
that $\lim_{\rho\to\infty} H_{\rhotrans}(X) \geq \log \card{\supp(P)}$; the
reverse inequality follows from~\eqref{eq:renyi_bounds}. 

As to the limit
when $\rho$ approaches zero, note that if $R> H(X)$, then the probability that
the cardinality of $f_n^{-1}(f_n(X^n))$ exceeds one can be driven to zero
exponentially fast in~$n$ \cite[Theorem~2.15]{csiszar2011information}, say as
$e^{-n\zeta}$ for some~$\zeta>0$ and sufficiently large $n$. 
And since $\card{f_n^{-1}(f_n(X^n))}$ is trivially upper-bounded
by $2^{n \log \card{\set{X}}}$, the $\rho$-th moment of
$\abs{f_n^{-1}(f_n(X^n))}$
will tend to one if $\rho<\zeta/\log\card{\set{X}}$. Thus, $\lim_{\rho\to 0}
H_{\rhotrans}(X) \leq  H(X)$ and, in view of~\eqref{eq:renyi_bounds},
\begin{equation}
\lim_{\rho\to 0} H_{\rhotrans}(X) =  H(X).
\end{equation}

The rest of this paper is organized as follows. 
In Section~\ref{sec:notation} we introduce some notation. 
In Section~\ref{sec:proof} we prove Theorem~\ref{thm:oneshot}. 
In Section~\ref{sec:mismatch} we consider a mismatched version of the direct
part of Theorem~\ref{thm:oneshot} (i.e., the upper bound), where $f$ is designed
based on the law $Q$ instead of $P$. We show that the penalty incurred by this
mismatch can be expressed in terms of a
divergence measure between $P$ and $Q$ that was proposed by
Sundaresan~\cite{sundaresan2007guessing}. 
In Section~\ref{sec:universal} we state and prove a universal version of 
the direct part of Theorem~\ref{thm:main} for IID sources.
In Section~\ref{sec:sideinfo} we generalize Theorems~\ref{thm:oneshot}
and~\ref{thm:main} to account for the presence of side-information, where the
key quantity is a conditional version of R\'enyi entropy. We also generalize the
result from Section~\ref{sec:universal}.
In Section~\ref{sec:lossy} we study a rate-distortion version of the problem
for IID sources, where the key quantity is ``R\'enyi's analog to the rate-distortion
function'' introduced by Arikan and Merhav~\cite{arikan1998guessing}. 
In Section~\ref{sec:costs} we study the problem of encoding IID tasks when 
different tasks may have different costs.

\section{Notation and Preliminaries}
\label{sec:notation}
We denote by $\naturals$ the set of positive integers. 
The cardinality of a finite set $\set{X}$ is denoted by $\card{\set{X}}$. 
We use the notation $x^n=(x_1,\ldots,x_n)$ for $n$-tuples.
If $P$ is a PMF on $\set{X}$, then $P^n$ denotes the product PMF on $\set{X}^n$ 
\begin{equation}
P^n(x^n) = \prod_{i=1}^n P(x_i),\quad x^n\in\set{X}^n.
\end{equation}
The support of $P$ is denoted by $\supp(P)$, so
\begin{equation}
\supp(P) = \bigl\{x\in\set{X}:P(x)>0\bigr\}.
\end{equation}
If $\set{A} \subseteq \set{X}$, then we write $P(\set{A})$ in lieu of $\sum_{x\in \set{A}}P(x)$. 
If  $W(\cdot|x)$ is a PMF on a finite set $\set{Y}$ for every $x\in\set{X}$
(i.e., a channel from $\set{X}$ to $\set{Y}$), then 
$P\circ W$ denotes the induced joint PMF on $\set{X}\times\set{Y}$
\begin{equation}
(P\circ W)(x,y) = P(x)W(y|x),\quad (x,y)\in\set{X}\times\set{Y},
\end{equation}
and $PW$ denotes the induced marginal PMF on $\set{Y}$
\begin{equation}
(PW)(y) = \sum_{x\in\set{X}} P(x)W(y|x),\quad y\in\set{Y}.
\end{equation}
The collection of all PMFs on $\set{X}$ is denoted by $\set{P}(\set{X})$. The
collection of all channels from $\set{X}$ to $\set{Y}$ is denoted by
$\set{P}(\set{Y}|\set{X})$. 

For information-theoretic quantities (entropy, relative entropy, mutual
information, etc.) we adopt the notation in~\cite{csiszar2011information}. 
We need basic results from the Method of Types as presented in~\cite[Chapter~2]{csiszar2011information}. 
The set of types of sequences in $\set{X}^n$ (i.e., the set of rational PMFs with denominator $n$) is denoted by
$\set{P}_n(\set{X})$. The set of all $x^n \in \set{X}^n$ of type~$Q\in\set{P}_n(\set{X})$ (i.e., the
type class of $Q$) is denoted by $T^{(n)}_Q$ or by $T_Q$ if $n$ is clear from the
context. The $V$-shell of a sequence $x^n \in \set{X}^n$ is denoted by $T_V(x^n)$. 

The ceiling of a real number $\xi$ (i.e., the smallest integer no smaller than
$\xi$) is denoted by $\lceil \xi\rceil$. We frequently use the inequality
\begin{equation}
\label{eq:ceiling}
\lceil\xi\rceil^\rho<1+2^\rho \xi^\rho,\quad \xi\geq0,
\end{equation}
which is easily checked by considering separately the cases~$0\leq\xi\leq 1$
and~$\xi>1$. As mentioned in the introduction, $\log(\cdot)$ denotes the base-2
logarithm, and $\log_\alpha(\cdot)$ denotes the base-$\alpha$ logarithm for
general $\alpha>1$. 

\section{Proof of Theorem~\ref{thm:oneshot}}
\label{sec:proof}
\subsection{The Lower Bound (Converse)}
\label{sec:lower_bound}
The proof of the lower bound~\eqref{eq:oneshot_lb} in Theorem~\ref{thm:oneshot} 
is inspired by the proof of~\cite[Theorem 1]{arikan1996inequality}.
It hinges on the following simple observation. 
\begin{proposition}
\label{prop:count_lists}
If $\set{L}_1,\ldots,\set{L}_M$ is a partition of a finite 
set~$\set{X}$ into $M$ nonempty subsets, i.e., 
\begin{equation}
\bigcup_{m=1}^M \set{L}_m = \set{X}\quad\text{and}\quad 
\text{($\set{L}_m\cap\set{L}_{m'} = \emptyset$ iff  $m'\neq
m$)}, 
\end{equation}
and $L(x)$ is the cardinality of the subset containing $x$, i.e.,
$L(x)=\card{\set{L}_m}$ if $x\in\set{L}_m$,  then
\begin{equation}
\label{eq:count_lists}
\sum_{x\in\set{X}} \frac{1}{L(x)} = M. 
\end{equation}
\end{proposition}
\begin{proof}
\begin{align}
\sum_{x\in\set{X}} \frac{1}{L(x)}&= \sum_{m=1}^M \sum_{x\in\set{L}_m}
\frac{1}{L(x)}\\
&=\sum_{m=1}^M \sum_{x\in\set{L}_m} \frac{1}{\card{\set{L}_m}}\\
&= M.
\end{align}
\end{proof}
To prove the lower bound in Theorem~\ref{thm:oneshot}, 
fix any $f\colon\set{X}\to\{1,\ldots,M\}$, 
and let $N$ denote the number of nonempty subsets in the partition
$f^{-1}(1),\ldots,f^{-1}(M)$. 
Note that for this partition the cardinality of the subset containing $x$
is 
\begin{equation}
\label{eq:induced_partition}
L(x) = \card{f^{-1}(f(x))},\quad x\in\set{X}.
\end{equation}
Recall H\"older's Inequality: If $a$ and $b$ are functions from $\set{X}$ into
the nonnegative reals, and $p$ and $q$ are real numbers larger than one
satisfying $1/p+1/q=1$, then 
\begin{equation}
\label{eq:hoelder}
\sum_{x\in\set{X}} a(x) b(x) \leq \biggl(\sum_{x\in\set{X}} a(x)^p\biggr)^{1/p}
\biggl(\sum_{x\in\set{X}} b(x)^q\biggr)^{1/q}.
\end{equation}
Rearranging~\eqref{eq:hoelder} gives
\begin{equation}
\label{eq:hoelder_alt}
\sum_{x\in\set{X}} a(x)^p \geq \biggl(\sum_{x\in\set{X}} b(x)^q\biggr)^{-p/q}
\biggl(\sum_{x\in\set{X}} a(x)b(x)\biggr)^{p}.
\end{equation}
Substituting $p=1+\rho$, $q=(1+\rho)/\rho$, $a(x) = P(x)^{\frac{1}{1+\rho}}
\card{f^{-1}(f(x))}^{\frac{\rho}{1+\rho}}$ and $b(x) =
\card{f^{-1}(f(x))}^{-\frac{\rho}{1+\rho}}$ in~\eqref{eq:hoelder_alt}, we obtain
\begin{align}
&\sum_{x\in\set{X}} P(x) \card{f^{-1}(f(x))}^{\rho}\notag\\
&\quad\geq \biggl(\sum_{x\in\set{X}} \frac{1}{\card{f^{-1}(f(x))}}\biggr)^{-\rho}
\biggl(\sum_{x\in\set{X}} P(x)^{\frac{1}{1+\rho}}
\biggr)^{1+\rho}\label{eq:lb_proof_501}\\
&\quad=2^{\rho(H_{\rhotrans}(X)-\log N)}\label{eq:440_43}\\
&\quad\geq  2^{\rho(H_{\rhotrans}(X)-\log M)},\label{eq:441_532}
\end{align}
where~\eqref{eq:440_43} follows from~\eqref{eq:renyi}, \eqref{eq:induced_partition}, and
Proposition~\ref{prop:count_lists}; and where~\eqref{eq:441_532} follows
because $N \leq M$.\qed

\subsection{The Upper Bound (Direct Part)}
\label{sec:upper_bound}
The key to the upper bound in Theorem~\ref{thm:oneshot} is the following 
reversed version of Proposition~\ref{prop:count_lists}; a proof is
provided in Appendix~\ref{sec:proposition}.
\begin{proposition}
\label{prop:sufficiency}
If $\set{X}$ is a finite set, $\lambda\colon \set{X} \to \naturals\cup\{+\infty\}$ and
\begin{equation}
\label{eq:sum_of_recip_equals_mu}
\sum_{x\in\set{X}} \frac{1}{\lambda(x)} = \mu
\end{equation}
(with the convention $1/\infty=0$), then there exists a partition of $\set{X}$ into at most
\begin{equation}
\label{eq:alpha}
\min_{\alpha >1} \lfloor\alpha \mu + \log_\alpha \card{\set{X}}+2\rfloor
\end{equation}
subsets such that
\begin{equation}
\label{eq:card_bound}
L(x) \leq \min\{\lambda(x),\card{\set{X}}\},\quad \text{for all $x\in\set{X}$,}
\end{equation}
where $L(x)$ is the cardinality of the subset containing $x$. 
\end{proposition}
(Proposition~\ref{prop:count_lists} cannot be fully reversed in
the sense that~\eqref{eq:alpha} cannot be replaced with $\mu$. 
Indeed, consider $\set{X}=\{a,b,c,d\}$ with $\lambda(a)=1$,
$\lambda(b)=2$, and $\lambda(c)=\lambda(d)=4$. In this example, $\mu$ equals~2, 
but we need~3 subsets to satisfy the cardinality constraints.)

Since H\"older's Inequality~\eqref{eq:hoelder} holds with equality if, and only
if, $a(x)^p$ is
proportional to $b(x)^q$, it follows that the
lower bound in Theorem~\ref{thm:oneshot} holds with equality 
if, and only if, $\card{f^{-1}(f(x))}$ is proportional to $P(x)^{-1/(1+\rho)}$.
We derive~\eqref{eq:oneshot_ub} by constructing a partition that approximately satisfies this
relationship.
To this end, we use Proposition~\ref{prop:sufficiency} with $\alpha=2$ in~\eqref{eq:alpha}
and
\begin{equation}
\label{eq:lambda}
\lambda(x) = \begin{cases} \bigl\lceil \beta\,
P(x)^{-\frac{1}{1+\rho}}\bigr\rceil& \text{if $P(x)>0$,}\\+\infty& \text{if
$P(x)=0$,}\end{cases}
\end{equation}
where we choose $\beta$ just large enough to guarantee the existence of a
partition of $\set{X}$ into at most $M$ subsets  
satisfying~\eqref{eq:card_bound}. For $M>\log \card{\set{X}}+2$ this is accomplished by the choice
\begin{equation}
\label{eq:beta}
\beta= \frac{2\sum_{x\in\set{X}}
P(x)^{\frac{1}{1+\rho}}}{M-\log\card{\set{X}}-2}.
\end{equation}
Indeed, with this choice
\begin{align}
\mu&=\sum_{x\in\set{X}} \frac{1}{\lambda(x)} \\
&\leq \sum_{x\in\set{X}} \frac{P(x)^{\frac{1}{1+\rho}}}{\beta} \\
&=\frac{M-\log\card{\set{X}}-2}{2},
\end{align}
and hence
\begin{equation}
2\mu + \log \card{\set{X}} +2 \leq M.
\end{equation}
Let then the partition
$\set{L}_1,\ldots,\set{L}_N$ with $N \leq M$ be as promised by
Proposition~\ref{prop:sufficiency}. 
Construct an encoder $f\colon\set{X}\to\{1,\ldots,M\}$ by setting $f(x) = m$ if $x\in
\set{L}_m$. For this encoder,
\begin{align}
&\sum_{x\in\set{X}} P(x)\card{f^{-1}(f(x))}^\rho\notag\\
&\quad= \sum_{x:P(x)>0} P(x)
L(x)^\rho\label{eq:370}\\
&\quad\leq \sum_{x:P(x)>0} P(x) \lambda(x)^\rho\label{eq:371}\\
&\quad=\sum_{x:P(x)>0} P(x) \bigl\lceil \beta\,
P(x)^{-\frac{1}{1+\rho}}\bigr\rceil^\rho\\
&\quad< 1+ (2\beta)^{\rho} \sum_{x:P(x)>0}
P(x)^{\frac{1}{1+\rho}}\label{eq:oneshot_proof_strict_eq}\\
&\quad=1+ 2^{\rho(H_{\rhotrans}(X)-\log \widetilde{M})},\label{eq:372}
\end{align}
where~\eqref{eq:oneshot_proof_strict_eq} follows from~\eqref{eq:ceiling}, 
and where~$\widetilde{M}$ is as in Theorem~\ref{thm:oneshot}.\qed

\section{Mismatch} 
\label{sec:mismatch}
The key to the upper bound in Theorem~\ref{thm:oneshot} was to use
Proposition~\ref{prop:sufficiency} with $\lambda$ as in~\eqref{eq:lambda}--\eqref{eq:beta} 
to obtain a partition of $\set{X}$ for which the cardinality of the
subset containing $x$ is approximately proportional to
$P(x)^{-1/(1+\rho)}$. Evidently, this construction requires knowledge of the
distribution~$P$ of~$X$. (But see Section~\ref{sec:universal} for a universal
version of the direct part of Theorem~\ref{thm:main} for IID sources that does not require
knowledge of the source's distribution.) 

In this section, we study the penalty when~$P$ is replaced with~$Q$
in~\eqref{eq:lambda} and~\eqref{eq:beta}. 
Since it is then still true that
\begin{equation}
\mu \leq \frac{M-\log \card{\set{X}}-2}{2}, 
\end{equation}
Proposition~\ref{prop:sufficiency} guarantees the existence of a partition 
of~$\set{X}$ into at most~$M$ subsets satisfying~\eqref{eq:card_bound}. Constructing 
an encoder $f$ from this partition as in Section~\ref{sec:upper_bound} 
and following steps similar to~\eqref{eq:370}--\eqref{eq:372} yields
\begin{multline}
\label{eq:mismatch_bound}
\sum_{x\in\set{X}} P(x)\card{f^{-1}(f(x))}^\rho\\
<1+2^{\rho(H_{\rhotrans}(X) +
\Delta_{\rhotrans}(P||Q) - \log \widetilde{M})},
\end{multline}
where $\widetilde{M}$ is as in Theorem~\ref{thm:oneshot}, and where
\begin{multline}
\label{eq:divergence}
\Delta_\alpha(P||Q)\\ \triangleq \log \frac{\sum_{x\in\set{X}}
Q(x)^\alpha}{\bigl(\sum_{x\in\set{X}} P(x)^\alpha\bigr)^{\frac{1}{1-\alpha}}}\biggl(\sum_{x\in\set{X}}
\frac{P(x)}{Q(x)^{1-\alpha}}\biggr)^{\frac{\alpha}{1-\alpha}}.
\end{multline}
The parameter $\alpha$ can be any positive number not equal to one. We use the convention $0/0=0$ and
$a/0=+\infty$ if $a>0$. Thus, $\Delta_{\rhotrans}(P||Q)<\infty$ only if the support of $P$ is contained in the
support of $Q$. 

The penalty in the exponent on the RHS of~\eqref{eq:mismatch_bound} when compared to the upper bound in
Theorem~\ref{thm:oneshot} is thus given by $\Delta_{\rhotrans}(P||Q)$. 
To reinforce this, note further that 
\begin{equation}
\label{eq:product_div}
\Delta_\alpha(P^n||Q^n) = n \Delta_\alpha(P||Q).
\end{equation}
Consequently, if the source $\{X_i\}_{i=1}^\infty$ is IID $P$ and we construct
$f_n\colon\set{X}^n\to\{1,\ldots,2^{nR}\}$ based on~$Q^n$ instead of~$P^n$, 
we obtain the bound
\begin{multline}
\label{eq:mismatch_upper_bound}
\bigE{\card{f_n^{-1}(f_n(X^n))}^\rho}\\
<1+2^{n\rho(H_{\rhotrans}(X_1) +
\Delta_{\rhotrans}(P||Q) - R + \delta_n)},
\end{multline}
where $\delta_n\to 0$ as $n\to\infty$. 
The RHS of~\eqref{eq:mismatch_upper_bound} tends to one provided that
$R>H_{\rhotrans}(X_1)+\Delta_{\rhotrans}(P||Q)$. 
Thus, in the IID case $\Delta_{\rhotrans}(P||Q)$
is the rate penalty incurred by the mismatch between $P$ and $Q$.

The family of divergence measures $\Delta_{\alpha}(P||Q)$ was first identified by
Sundaresan~\cite{sundaresan2007guessing} who showed that it plays a similar
role in the Massey-Arikan guessing problem~\cite{massey1994guessing, arikan1996inequality}. 
We conclude this section with some properties of~$\Delta_{\alpha}(P||Q)$. 
Properties 1--3 (see below) were given in~\cite{sundaresan2007guessing}; we
repeat them here for completeness.
Note that R\'enyi's divergence (see, e.g.,~\cite{csiszar1995generalized})
\begin{equation}
D_\alpha(P||Q) = \frac{1}{\alpha-1} \log \sum_{x\in\set{X}} P(x)^\alpha
Q(x)^{1-\alpha},
\end{equation}
satisfies Properties 1 and 3 but none of the others in general.
\begin{proposition}
\label{prop:properties}
The functional $\Delta_\alpha(P||Q)$ has the following properties.
\begin{enumerate}
\item $\Delta_\alpha(P||Q) \geq 0$ with equality if, and only if, $P=Q$. 
\item $\Delta_\alpha(P||Q)=\infty$ if, and only if, $\supp(P)\not\subseteq
\supp(Q)$ or ($\alpha>1$ and $\supp(P)\cap \supp(Q) = \emptyset$.) 
\item $\lim_{\alpha\to 1} \Delta_\alpha(P||Q) = D(P||Q)$. 
\item $\lim_{\alpha\to 0}\Delta_\alpha(P||Q) = \log
\frac{\card{\supp(Q)}}{\card{\supp(P)}}$ if $\supp(P) \subseteq \supp(Q)$. 
\item $\lim_{\alpha\to \infty} \Delta_\alpha(P||Q) = \log
\frac{\max_{x\in\set{X}}
P(x)}{\frac{1}{\card{\set{Q}}}\sum_{x\in\set{Q}}P(x)}$, where 
$\set{Q}=\bigl\{x\in \set{X}: Q(x) = \max_{x'\in\set{X}} Q(x')\bigr\}$. 
\end{enumerate}
\end{proposition}
\begin{proof}
Property 2 follows by inspection of~\eqref{eq:divergence}. Properties 3--5 follow by
simple calculus. As to Property 1, 
consider first the case where $0<\alpha<1$. In view of Property 2, we may assume
that $\supp(P) \subseteq \supp(Q)$. H\"older's Inequality~\eqref{eq:hoelder} with
$p=1/\alpha$ and $q=1/(1-\alpha)$ gives
\begin{align}
&\sum_{x\in\set{X}} P(x)^\alpha\notag\\
&\quad=\sum_{x\in\supp(P)}
\frac{P(x)^\alpha}{Q(x)^{\alpha(1-\alpha)}} Q(x)^{\alpha(1-\alpha)}\\
&\quad\leq \biggl(\sum_{x\in\supp(P)}
 \frac{P(x)}{Q(x)^{1-\alpha}}\biggr)^\alpha \biggl(\sum_{x\in\supp(P)} Q(x)^\alpha
\biggr)^{1-\alpha}\\
&\quad\leq \biggl(\sum_{x\in\set{X}}
 \frac{P(x)}{Q(x)^{1-\alpha}}\biggr)^\alpha \biggl(\sum_{x\in\set{X}} Q(x)^\alpha
\biggr)^{1-\alpha}.\label{eq:apply_holder_651}
\end{align}
Dividing by $\sum_{x}P(x)^\alpha$ and 
taking $(1-\alpha$)-th roots shows that $\Delta_\alpha(P||Q)\geq 0$. 
The condition for equality in H\"older's Inequality implies that 
equality holds if, and only if, $P=Q$. Consider next the 
case where~$\alpha>1$. We may assume $\supp(P)\cap \supp(Q) \neq \emptyset$
(Property~2). H\"older's Inequality with $p=\alpha$ and
$q=\alpha/(\alpha-1)$ gives
\begin{align}
\sum_{x\in\set{X}}\frac{P(x)}{Q(x)^{1-\alpha}}&=\sum_{x\in\set{X}}P(x)
Q(x)^{\alpha-1}\\
&\leq \biggl(\sum_{x\in\set{X}} P(x)^\alpha\biggr)^{\frac{1}{\alpha}}
\biggl(\sum_{x\in\set{X}}
Q(x)^{\alpha}\biggr)^{\frac{\alpha-1}{\alpha}}.\label{eq:apply_holder_664}
\end{align}
Dividing by $\sum_{x} P(x)/Q(x)^{1-\alpha}$ and raising to the
power of $\alpha/(\alpha-1)$ shows that $\Delta_{\alpha}(P||Q)\geq 0$. 
Equality holds if, and only if, $P=Q$. 
\end{proof}

\section{Universal Encoders for IID Sources}
\label{sec:universal}
In Section~\ref{sec:introduction} the direct part of Theorem~\ref{thm:main} is
proved using the upper bound in Theorem~\ref{thm:oneshot}.
It is pointed out in Section~\ref{sec:mismatch} that the
construction of the encoder in the proof of this upper bound
requires knowledge of the distribution of $X$. 
As the next result shows, however, for IID sources we do not need to know the
distribution of the source to construct good encoders.
\begin{theorem}
\label{thm:universal}
Let $\set{X}$ be a finite set, and let $\rho>0$. For every rate $R>0$ 
there exist encoders $f_n\colon \set{X} \to \{1,\ldots,
2^{nR}\}$ such that for every IID source~$\{X_i\}_{i=1}^\infty$ with alphabet $\set{X}$, 
\begin{equation}
\bigE{\card{f^{-1}_n(f_n(X^n))}^\rho} <
1+2^{-n\rho(R-H_{\rhotrans}(X_1)-\delta_n)},
\end{equation}
where
\begin{equation}
\delta_n = \frac{1+(1+\rho^{-1})\card{\set{X}}\log(n+1)}{n}.
\end{equation}
In particular, 
\begin{equation}
\lim_{n\to\infty}\bigE{\card{f^{-1}_n(f_n(X^n))}^\rho}=1,
\end{equation}
whenever $H_{\rhotrans}(X_1) < R$. 
\end{theorem}
\begin{proof}
We first partition $\set{X}^n$ into the different type classes $T_Q$, of which there
are less than $(n+1)^{\card{\set{X}}}$. We then partition each $T_Q$ into
$2^{n(R-\delta'_n)}$ subsets of cardinality at most $\lceil \card{T_Q}
2^{-n(R-\delta'_n)}\rceil$ where $\delta'_n = n^{-1}
\card{\set{X}} \log(n+1)$. Since $\card{T_Q} \leq 2^{nH(Q)}$, each~$x^n\in T_Q$
thus ends up in a subset of cardinality at most
\begin{equation}
\bigl\lceil 2^{n(H(Q)-R+\delta'_n)}\bigr\rceil.
\end{equation}
Note that the total number of subsets in the partition does not exceed $2^{nR}$. 
We construct $f_n\colon \set{X} \to \{1,\ldots, 2^{nR}\}$ by
enumerating the subsets in the partition with the numbers in $\{1,\ldots, 2^{nR}\}$ and by mapping 
to~$m\in\{1,\ldots,2^{nR}\}$ the $x^n$'s that comprise the $m$-th subset. 
Suppose now that $\{X_i\}_{i=1}^\infty$ is IID $P$ and observe that
\begin{align}
&\bigE{\card{f_n^{-1}(f_n(X^n))}^\rho}\notag\\
&\quad=\sum_{x^n \in \set{X}^n} P^n(x^n)
\card{f_n^{-1}(f_n(X^n))}^\rho\\
&\quad\leq  \sum_{Q\in \set{P}_n(\set{X})} \sum_{x^n \in T_Q} P^n(x^n) \bigl\lceil
2^{n(H(Q)-R+\delta'_n)}\bigr\rceil^\rho\label{eq:iid_ub_card}\\
&\quad< 1+ 2^{\rho}\sum_{Q\in \set{P}_n(\set{X})}2^{n\rho(H(Q)-R+\delta'_n)}
\sum_{x^n \in T_Q} P^n(x^n)\label{eq:iid_ub_ceil}\\
&\quad\leq 1+ 2^\rho \sum_{Q\in \set{P}_n(\set{X})} 
2^{-n\rho(R-H(Q)+\rho^{-1} D(Q||P)-\delta'_n)}\label{eq:iid_ub_type}\\
&\quad\leq 1+ 2^{-n\rho(R-H_{\rhotrans}(X_1)-\delta_n)}.\label{eq:iid_ub_poly_types}
\end{align}
Here~\eqref{eq:iid_ub_card} follows from the construction of $f_n$;
\eqref{eq:iid_ub_ceil} follows from~\eqref{eq:ceiling};
\eqref{eq:iid_ub_type} follows because the probability of the source
emitting a sequence of type $Q$ is at most~$2^{-nD(Q||P)}$;
and~\eqref{eq:iid_ub_poly_types} follows from the identity (see~\cite{arikan1996inequality})
\begin{equation}
\label{eq:renyi_variation}
H_{\rhotrans}(X_1) = \max_{Q\in\set{P}(\set{X})} H(Q) - \rho^{-1}
D(Q||P),
\end{equation} 
and the fact that $\card{\set{P}_n(\set{X})} < (n+1)^{\card{\set{X}}}$. 
\end{proof}

\section{Tasks with Side-Information}
\label{sec:sideinfo}
In this section we generalize Theorems~\ref{thm:oneshot}, \ref{thm:main},
and~\ref{thm:universal} to
account for side-information: A task $X$ and side-information
$Y$ are drawn according to a joint PMF $P_{X,Y}$ on $\set{X}\times\set{Y}$, where 
both~$\set{X}$ and~$\set{Y}$ are finite, and where the side-information is
available to both the task describer (encoder) and the tasks performer. 
The encoder is now of the form
\begin{equation}
f\colon \set{X}\times\set{Y}\to \{1,\ldots,M\}.
\end{equation}
If the realization of $(X,Y)$ is $(x,y)$ and $f(x,y)=m$, then all the tasks in the set
\begin{equation}
f^{-1}(m,y) \eqdef \{x\in \set{X}: f(x,y)=m\}
\end{equation}
are performed. 
As in Section~\ref{sec:introduction}, we seek to minimize for a given $M$ the $\rho$-th moment of the number of
performed tasks
\begin{multline}
\bigE{\card{f^{-1}(f(X,Y),Y)}^\rho}\\
= \sum_{x\in\set{X}}\sum_{y\in\set{Y}}
P_{X,Y}(x,y) \card{f^{-1}(f(x,y),y)}^{\rho}.
\end{multline}
The key quantity here is a conditional version of R\'enyi entropy (proposed by
Arimoto~\cite{arimoto1977information}):
\begin{equation}
H_{\rhotrans}(X|Y) = \frac{1}{\rho} \log \sum_{y\in\set{Y}}
\biggl(\sum_{x\in\set{X}} P_{X,Y}(x,y)^{\frac{1}{1+\rho}} \biggr)^{1+\rho}.
\end{equation}
Theorem~\ref{thm:oneshot} can be generalized as follows.
\begin{theorem}
\label{thm:oneshot_sideinfo}
Let $(X,Y)$ be a pair of chance variables taking value in the finite set
$\set{X}\times\set{Y}$, and let $\rho>0$. 
\begin{enumerate}
\item For all positive integers $M$ and every $f\colon\set{X}\times\set{Y}\to \{1,\ldots,M\}$, 
\begin{equation}
\label{eq:oneshot_sideinfo_lb}
\bigE{\card{f^{-1}(f(X,Y),Y)}^\rho}\geq  2^{\rho(H_{\rhotrans}(X|Y)-\log M)}. 
\end{equation}
\item For all integers $M > \log\card{\set{X}}+2$ there exists $f\colon
\set{X}\times\set{Y}\to\{1,\ldots,M\}$ such that
\begin{equation}
\label{eq:oneshot_sideinfo_ub}
\bigE{\card{f^{-1}(f(X,Y),Y)}^\rho}
<1+2^{\rho(H_{\rhotrans}(X|Y)-\log \widetilde{M})},
\end{equation}
where $\widetilde{M}=(M-\log\card{\set{X}}-2)/4$.
\end{enumerate}
\end{theorem}
As a corollary we obtain a generalization of Theorem~\ref{thm:main}. 
\begin{theorem}
\label{thm:main_si}
Let $\{(X_i,Y_i)\}_{i=1}^\infty$ be any source with finite alphabet
$\set{X}\times\set{Y}$, and let $\rho>0$.
\begin{enumerate}
\item If $R> \limsup_{n\to\infty} H_{\rhotrans}(X^n|Y^n)/n$, then there exist
$f_n \colon \set{X}^n\times\set{Y}^n \to \{1,\ldots, 2^{nR}\}$ such that
\begin{equation}
\lim_{n\to\infty} \bigE{\card{f^{-1}_n(f_n(X^n,Y^n),Y^n)}^\rho} = 1.
\end{equation}
\item If $R< \liminf_{n\to\infty} H_{\rhotrans}(X^n|Y^n)/n$, then for any choice
of $f_n \colon \set{X}^n\times\set{Y}^n \to \{1,\ldots, 2^{nR}\}$
\begin{equation}
\lim_{n\to\infty} \bigE{\card{f^{-1}_n(f_n(X^n,Y^n),Y^n)}^\rho} = \infty.
\end{equation}
\end{enumerate}
\end{theorem}
To prove~\eqref{eq:oneshot_sideinfo_lb} fix $M$ and $f\colon
\set{X}\times\set{Y}\to \{1,\ldots, M\}$.
Note that for every $y\in \set{Y}$ the sets
$f^{-1}(1,y),\ldots,f^{-1}(M,y)$ form a partition of $\set{X}$, and the
cardinality of the subset containing $x$ is~$\card{f^{-1}(f(x,y),y)}$. 
Following steps similar to~\eqref{eq:lb_proof_501}--\eqref{eq:441_532}, we obtain
\begin{multline}
\sum_{x\in \set{X}} P_{X|Y}(x|y) \card{f^{-1}(f(x,y),y)}^\rho\\
\geq
2^{-\rho \log M} \biggl(\sum_{x\in\set{X}}
P_{X|Y}(x|y)^{\frac{1}{1+\rho}}\biggr)^{1+\rho},\quad y\in\set{Y}.
\end{multline}
Multiplying both sides by $P_Y(y)$ and summing over all $y\in \set{Y}$
establishes~\eqref{eq:oneshot_sideinfo_lb}.

To prove~\eqref{eq:oneshot_sideinfo_ub} fix some $y\in\set{Y}$
and replace $P(x)$ with $P_{X|Y}(x|y)$ everywhere in the proof of the upper
bound in Theorem~\ref{thm:oneshot} (see Section~\ref{sec:upper_bound}) to obtain 
an encoder $f_y \colon \set{X} \to \{1,\ldots, M\}$ satisfying
\begin{multline}
\label{eq:ub_817}
\sum_{x\in\set{X}} P_{X|Y}(x|y)
\card{f_y^{-1}(f_y(x))}^{\rho}\\
< 1+ 2^{-\rho \log \widetilde{M}}
\biggl(\sum_{x\in\set{X}} P_{X|Y}(x|y)^{\frac{1}{1+\rho}}\biggr)^{1+\rho}.
\end{multline}
Setting $f(x,y)=f_y(x)$, multiplying both sides of~\eqref{eq:ub_817} 
by~$P_Y(y)$, and summing over all $y\in\set{Y}$
establishes~\eqref{eq:oneshot_sideinfo_ub}.\qed

One may also generalize Theorem~\ref{thm:universal}:
\begin{theorem}
\label{thm:universal_si}
Let $\set{X}$ and $\set{Y}$ be finite sets, and let $\rho>0$. 
For every rate $R>0$ there exist encoders $f_n\colon \set{X}\times\set{Y} \to \{1,\ldots,
2^{nR}\}$ such that for every IID source~$\{(X_i,Y_i)\}_{i=1}^\infty$ with alphabet $\set{X}\times\set{Y}$, 
\begin{multline}
\label{eq:universal_si_884}
\bigE{\card{f^{-1}_n(f_n(X^n,Y^n),Y^n)}^\rho}\\
< 1+2^{-n\rho(R-H_{\rhotrans}(X_1|Y_1)-\delta_n)},
\end{multline}
where
\begin{equation}
\delta_n = \frac{1+(1+\rho^{-1})\card{\set{X}}\card{\set{Y}}\log(n+1)+\rho^{-1}
\card{\set{X}}\log(n+1)}{n}.
\end{equation}
In particular, 
\begin{equation}
\lim_{n\to\infty}\bigE{\card{f^{-1}_n(f_n(X^n,Y^n),Y^n)}^\rho}=1,
\end{equation}
whenever $H_{\rhotrans}(X_1|Y_1) < R$. 
\end{theorem}
\begin{proof}
We fix an arbitrary $y^n \in \set{Y}^n$ and partition $\set{X}^n$ into the different
$V$-shells $T_V(y^n)$ (see~\cite[Chapter~2]{csiszar2011information}) of which
there are less than $(n+1)^{\card{\set{X}}\card{\set{Y}}}$. We then partition each $V$-shell into
$2^{n(R-\delta_n')}$ subsets of cardinality at most 
$\lceil \card{T_V(y^n)} 2^{-n(R-\delta'_n)}\rceil$ where $\delta_n' = n^{-1} \card{\set{X}}\card{\set{Y}} \log(n+1)$. Since
$\card{T_V(y^n)}\leq  2^{n H(V|P_{y^n})}$, where $P_{y^n}$ denotes the type of
$y^n$, each~$x^n \in T_V(y^n)$ will end up in a subset of cardinality at most 
\begin{equation}
\bigl\lceil 2^{n(H(V|P_{y^n})-R+\delta_n')}\bigr\rceil.
\end{equation}
From this partition we construct $f_n(\cdot, y^n) \colon \set{X} \to \{1,\ldots, 2^{nR}\}$ 
by enumerating the subsets with the numbers $1$ through $2^{nR}$
and by mapping to each $m \in \{1,\ldots, 2^{nR}\}$ the $x^n$'s that comprise
the $m$-th subset. Carrying out this construction for every $y^n \in \set{Y}^n$
yields an encoder $f\colon \set{X}\times\set{Y}\to \{1,\ldots, 2^{nR}\}$. 
Suppose now that~$\{(X_i,Y_i)\}_{i=1}^\infty$ is IID $P_{X,Y}$
and observe that for every $y^n\in\set{Y}^n$ with $P_Y^{(n)}(y^n)>0$, 
\begin{align}
&\sum_{x^n \in \set{X}^n} P^{(n)}_{X|Y}(x^n|y^n)
\bigcard{f_n^{-1}\bigl(f_n(x^n,y^n),y^n\bigr)}^\rho\notag\\
&\leq \!\!\!\sum_{V:T_V(y^n)\neq \emptyset}
\sum_{x^n \in T_V(y^n)} P^{(n)}_{X|Y}(x^n|y^n) \bigl\lceil
2^{n(H(V|P_{y^n})-R+\delta_n')} \rceil^\rho\label{eq:universal_si_919}\\
&<1+2^\rho \!\!\!\!\!\!\sum_{V:T_V(y^n) \neq \emptyset}\!\!\!
2^{n\rho(H(V|P_{y^n})-R+\delta_n')}\!\!\!\!\!\! \sum_{x^n \in T_V(y^n)} P^{(n)}_{X|Y}(x^n|y^n)\label{eq:universal_si_923}\\
&<1+ 2^{\rho} \!\!\!\!\!\!\sum_{V:T_V(y^n) \neq \emptyset} 2^{-n
D(V||P_{X|Y}|P_{y^n})}2^{n\rho(H(V|P_{y^n})-R+\delta_n')}.
\label{eq:universal_si_921}
\end{align}
Here~\eqref{eq:universal_si_919} follows from the construction of $f_n$;
\eqref{eq:universal_si_923} follows from~\eqref{eq:ceiling};
\eqref{eq:universal_si_921} follows because conditional on $Y^n=y^n$ the
probability that $X^n$ is in the $V$-shell of $y^n$ is at most
$2^{-nD(V||P_{X|Y}|P_{y^n})}$. Noting that whether $T_V(y^n)$ is nonempty depends
on $y^n$ only via its type, it follows that the sum
in~\eqref{eq:universal_si_921} depends on $y^n$ only via $P_{y^n}$. Noting
further that the probability that $Y^n$ is of type $Q \in \set{P}_n(\set{Y})$ is
at most~$2^{-nD(Q||P_Y)}$ it follows
from~\eqref{eq:universal_si_919}--\eqref{eq:universal_si_921} upon
taking expectation with respect to $Y^n$ that
\begin{align}
&\bigE{\card{f^{-1}_n(f_n(X^n,Y^n),Y^n)}^\rho}< 1+ 2^{\rho}\!\!\! \sum_{Q \in \set{P}_n(\set{Y})}2^{-nD(Q||P_{Y})}\notag\\
&\qquad\times\sum_{V} 2^{-n
D(V||P_{X|Y}|P_{y^n})}2^{n\rho(H(V|P_{y^n})-R+\delta_n')},
\label{eq:universal_si_937}
\end{align}
where for a given $Q\in \set{P}_n(\set{Y})$ 
the inner sum extends over all~$V$ such that $T_V(y^n)$ is not empty for some
(and hence all)~$y^n$ of type $Q$. In Appendix~\ref{sec:cond_renyi_identity} it
is shown that
\begin{equation}
\label{eq:cond_renyi_identity}
H_{\rhotrans}(X_1|Y_1) = \max_{\substack{Q\in\set{P}(\set{Y})\\V\in
\set{P}(\set{X}|\set{Y})}} H(V|Q) - \rho^{-1} D(Q\circ V||P_{X,Y}).
\end{equation}
Using~\eqref{eq:cond_renyi_identity}, the identity
\begin{equation}
D(Q\circ V||P_{X,Y}) = D(Q||P_Y) + D(V||P_{X|Y}|Q),
\end{equation}
and the fact that the number of types of sequences in $\set{Y}^n$ is less than
$(n+1)^{\card{\set{Y}}}$ and the number of conditional types $V$ is less than
$(n+1)^{\card{\set{X}}\card{\set{Y}}}$, it follows that 
the RHS of~\eqref{eq:universal_si_937} is upper-bounded by the RHS of~\eqref{eq:universal_si_884}.
\end{proof}

\section{Coding for Tasks with a Fidelity Criterion}
\label{sec:lossy}
In this section we study a rate-distortion version of the problem described in
Section~\ref{sec:introduction}. We only treat IID sources and single-letter
distortion functions. Suppose that the source $\{X_i\}_{i=1}^\infty$ 
generates tasks from a finite set of tasks $\set{X}$ IID according to $P$.
Let $\hat{\set{X}}$ be some other finite set of tasks, and let $d\colon
\set{X}\times\hat{\set{X}} \to [0,\infty)$ be a function that measures the
dissimilarity, or distortion, between any pair of tasks in $\set{X} \times
\hat{\set{X}}$. The distortion function $d$ extends to $n$-tuples of tasks in the usual way:
\begin{equation}
d(x^n,\hat{x}^n) = \frac{1}{n}\sum_{i=1}^n d(x_i,\hat{x}_i),\quad
(x^n,\hat{x}^n) \in \set{X}^n\times \hat{\set{X}}^n.
\end{equation}
We assume that for every $x\in \set{X}$ there is some $\hat{x} \in \hat{\set{X}}$
for which $d(x,\hat{x})=0$, i.e.,  
\begin{equation}
\min_{\hat{x}\in\hat{\set{X}}} d(x,\hat{x}) = 0,\quad x\in \set{X}.
\end{equation}
We describe the first~$n$ tasks $X^n$ using~$nR$ bits with an encoder
\begin{equation}
\label{eq:lossy_encoder}
f\colon \set{X}^n \to \{1,\ldots, 2^{nR}\}.
\end{equation}
Subsequently, the description $f(X^n)$ of $X^n$ is decoded into a subset of
$\hat{\set{X}}^n$ by a decoder 
\begin{equation}
\label{eq:lossy_decoder}
\varphi\colon \{1,\ldots, 2^{nR}\} \to 2^{\hat{\set{X}}^n},
\end{equation}
where $2^{\hat{\set{X}}^n}$ denotes the collection of all subsets of $\hat{\set{X}}^n$. 
We require that the subset produced by the decoder always contain at least one $n$-tuple of tasks within distortion $D$ of
$X^n$,
i.e., we require
\begin{equation}
\label{eq:distortion}
\min_{\hat{x}^n \in \varphi(f(x^n))} d(x^n,\hat{x}^n) \leq D,\quad x^n\in\set{X}.
\end{equation}
All $n$-tuples of tasks in the set $\varphi(f(X^n))$ are performed. 
The next theorem shows that the infimum of all rates~$R$ for which the
$\rho$-th moment of the ratio of performed tasks to $n$ can be driven to one as
$n$ tends to infinity subject to the constraint~\eqref{eq:distortion} is given by
\begin{equation}
\label{eq:renyi_rd}
R_\rho(P,D) \triangleq \max_{Q\in\set{P}(\set{X})} R(Q,D)-\rho^{-1} D(Q||P),
\end{equation}
where $R(Q,D)$ is the classical rate-distortion function (see, e.g.,
\cite[Chapter~7]{csiszar2011information}) evaluated at the distortion level
$D$ for an IID Q source and distortion function $d$. The function~$R_\rho(P,D)$
(multiplied by $\rho)$ has previously appeared
in~\cite{arikan1998guessing} in the context of guessing.
\begin{theorem}
\label{thm:lossy}
Let $\{X_i\}_{i=1}^\infty$ be an IID $P$ source with finite alphabet $\set{X}$,
and let $D \geq 0$ and $\rho>0$.
\begin{enumerate}
\item If $R> R_\rho(P,D)$, then there exist $(f_n,\varphi_n)$ as
in~\eqref{eq:lossy_encoder} and~\eqref{eq:lossy_decoder}
satisfying~\eqref{eq:distortion} such that
\begin{equation}
\lim_{n\to\infty} \bigE{\card{\varphi_n(f_n(X^n))}^\rho} =1.
\end{equation}
\item If $R< R_\rho(P,D)$, then for any $(f_n,\varphi_n)$ 
as in~\eqref{eq:lossy_encoder} and~\eqref{eq:lossy_decoder}
satisfying~\eqref{eq:distortion}, 
\begin{equation}
\label{eq:lossy_explode}
\lim_{n\to\infty} \bigE{\card{\varphi_n(f_n(X^n))}^\rho} =\infty.
\end{equation}
\end{enumerate}
\end{theorem}
It follows immediately from Theorem~\ref{thm:lossy} that $R_\rho(P,D)$ is nonnegative 
and nondecreasing in $\rho>0$. Some other properties are
(see~\cite{arikan1998guessing} for proofs):
\begin{enumerate}
\item $R_\rho(P,D)$ is nonincreasing, continuous and convex in $D\geq 0$.
\label{item:continuous}
\item\label{item:dzero} $R_\rho(P,0) = H_{\rhotrans}(P)$.
\item $\lim_{\rho \to 0} R_\rho(P,D) = R(P,D)$. 
\item $\lim_{\rho \to\infty} R_\rho(P,D) = \max_{Q\in\set{P}(\set{X})} R(Q,D)$.
\end{enumerate}
The fact that $R_{\rho}(P,D)$ is a continuous function of $D$
(Property~\ref{item:continuous}) allows us to strengthen the converse statement
in Theorem~\ref{thm:lossy} as follows. Suppose that for every positive integer $n$ the
encoder/decoder pair $(f_n,\varphi_n)$ is as in~\eqref{eq:lossy_encoder}
and~\eqref{eq:lossy_decoder} with $R<R_\rho(P,D)$ and
satisfies~\eqref{eq:distortion} for some $D_n$ such that $\limsup_{n\to\infty}
D_n \leq D$. Then~\eqref{eq:lossy_explode} holds. Indeed, continuity implies that
$R<R_\rho(P,D+\epsilon)$ for a sufficiently small $\epsilon>0$, and $\limsup_{n\to\infty} D_n
\leq D$ implies that $D_n \leq D+\epsilon$ for all sufficiently large $n$. The
claim thus follows from the converse part of Theorem~\ref{thm:lossy} with $D$
replaced by $D+\epsilon$. 

Considering Property~\ref{item:dzero}, Theorem~\ref{thm:main} particularized to
IID sources can be recovered
from Theorem~\ref{thm:lossy} by taking $\hat{\set{X}} =
\set{X}$ and the Hamming distortion function 
\begin{equation}
\label{eq:hamming}
d(x,\hat{x}) = \begin{cases} 0& \text{if
$x=\hat{x}$,}\\1&\text{otherwise.}\end{cases}
\end{equation}

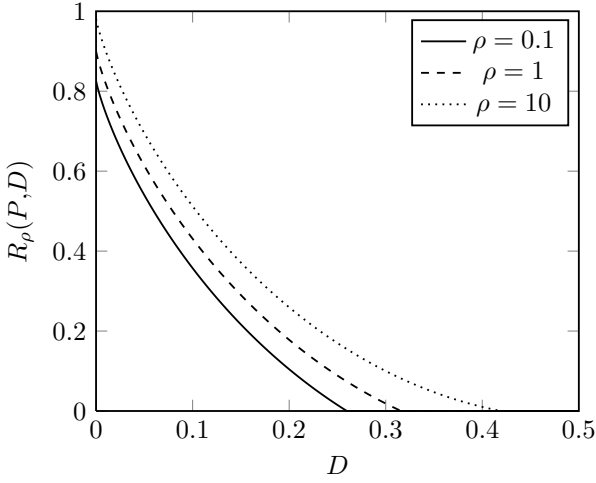
\begin{figure}
\centering
\begin{tikzpicture}
	\begin{axis}[xlabel=$D$, ylabel=$R_\rho(P{,}D)$, xmin=0, ymin=0, xmax=0.5,
	ymax=1, cycle list name=linestyles]
	\addplot table[x=D,y=rho_small] {R_plot3.dat};
	\addlegendentry{$\rho=0.1$}
	\addplot table[x=D,y=rho_medium] {R_plot3.dat};
	\addlegendentry{$\rho=1$}
	\addplot table[x=D,y=rho_large] {R_plot3.dat};
	\addlegendentry{$\rho=10$}
	\end{axis}
\end{tikzpicture}
\caption{$R_\rho(P,D)$ in bits for an IID Bernoulli-$(1/4)$ source and Hamming distortion}
\label{fig:plot}
\end{figure}
It was noted in~\cite{arikan1998guessing} that~$R_\rho(P,D)$ can be expressed in
closed form for binary sources and Hamming distortion:
\begin{proposition}
\label{prop:binary_hamming}
If $\set{X}=\hat{\set{X}} = \{0,1\}$, $d$ is the Hamming distortion
function~\eqref{eq:hamming},
and $P(0)=1-P(1)=p$, then 
\begin{equation*}
R_\rho(P,D) = \begin{cases} H_{\rhotrans}(p)-h(D)& \text{if $0\leq D
<h^{-1}\bigl(H_{\rhotrans}(p)\bigr)$,}\\ 0 &\text{if $D\geq
h^{-1}\bigl(H_{\rhotrans}(p)\bigr)$,}\end{cases}
\end{equation*}
where $h^{-1}(\cdot)$ denotes the inverse of the binary entropy 
function~$h(\cdot)$ on the interval $[0,1/2]$ and, with slight abuse of
notation, 
$H_{\rhotrans}(p) = H_{\rhotrans}(P)$. 
\end{proposition}
For a proof of Proposition~\ref{prop:binary_hamming} see \cite[Thereom
3]{arikan1998guessing} and subsequent remarks. 
A plot of~$R_\rho(D)$ for~$p=1/4$ and different values of~$\rho$ is shown in Figure~\ref{fig:plot}.

We now prove the direct part of Theorem~\ref{thm:lossy}. Fix $D \geq 0$ and
select an arbitrary $\delta>0$.
According to the Type Covering Lemma~\cite[Lemma~9.1]{csiszar2011information},
there is a positive integer $n(\delta)$ such that for all $n \geq n(\delta)$ and every type~$Q\in
\set{P}_n(\set{X})$ we can find a set $B^{(n)}_Q \subset \hat{\set{X}}^n$ of
cardinality at most $2^{n(R(Q,D)+\delta)}$ that covers $T^{(n)}_Q$ 
in the sense that for every $x^n \in T^{(n)}_Q$ there is
at least one $\hat{x}^n
\in B^{(n)}_Q$ with $d(x^n,\hat{x}^n) \leq D$. 
We henceforth assume that~$n \geq n(\delta)$. For each type $Q\in \set{P}_n(\set{X})$ we partition $B_Q^{(n)}$
into $2^{n(R-\delta_n)}$ subsets of cardinality at most
\begin{equation}
\bigl\lceil 2^{n(R(Q,D)+\delta-R+\delta_n)}\bigr\rceil, 
\end{equation}
where $\delta_n= n^{-1} \card{\set{X}} \log(n+1)$. Since the total number of
types is less than $(n+1)^{\card{\set{X}}}$, we can enumerate all the subsets of
all the different $B_Q^{(n)}$'s with the numbers $1,\ldots, 2^{nR}$. Let
$\varphi_n\colon \{1,\ldots, 2^{nR}\} \to 2^{\hat{\set{X}}^n}$ 
be the mapping that maps the index to the corresponding subset. (If there are
less than $2^{nR}$ subsets in our construction, then we map the remaining indices
to, say, the empty set.) 
We then construct~$f_n\colon \set{X}^n \to \{1,\ldots, 2^{nR}\}$ by mapping each
$x^n\in \set{X}^n$ of type $Q$ to an index of a subset of $B_Q^{(n)}$ that
contains an~$\hat{x}^n$ with $d(x^n,\hat{x}^n)\leq D$. Note that the encoder/decoder
pair thus constructed satisfies~\eqref{eq:distortion}, and
\begin{align}
&\bigE{\card{\varphi_n(f_n(X^n))}^\rho}\notag\\
&=\sum_{x^n \in \set{X}^n} P^n(x^n)
\card{\varphi_n(f_n(x^n))}^\rho\\
&\leq \sum_{Q \in \set{P}_n(\set{X})} \sum_{x^n \in T^{(n)}_Q} P^n(x^n) \bigl\lceil
2^{n(R(Q,D)+\delta-R+\delta_n)}\bigr\rceil^\rho\label{eq:lossy_card}\\
&<1+ 2^\rho\!\!\!\!\!\!\sum_{Q \in
\set{P}_n(\set{X})}2^{n\rho(R(Q,D)+\delta-R+\delta_n)}\!\!\!\sum_{x^n \in T^{(n)}_Q}
P^n(x^n)\label{eq:lossy_ceiling}\\
&<1+2^{\rho}\!\!\!\!\!\!\sum_{Q \in \set{P}_n(\set{X})}
2^{-n\rho(R+\rho^{-1}D(Q||P)-R(Q,D)-\delta-\delta_n)}\label{eq:lossy_type}\\
&\leq 1+2^{-n\rho(R-R_\rho(P,D)-\delta-\delta_n')},\label{eq:lossy_types}
\end{align}
where 
\begin{equation}
\delta_n'= \frac{1+(1+\rho^{-1})\card{\set{X}}\log(n+1)}{n}.
\end{equation}
Here \eqref{eq:lossy_card} follows from the construction of $f_n$ and $\varphi_n$;
\eqref{eq:lossy_ceiling} follows
from~\eqref{eq:ceiling};
\eqref{eq:lossy_type} follows because the probability of an IID $P$ source emitting a sequence of type $Q$ is
at most $2^{-nD(Q||P)}$;  and~\eqref{eq:lossy_types} follows from the definition
of $R_\rho(P,D)$ in~\eqref{eq:renyi_rd} and the fact that
$\card{\set{P}_n(\set{X})}< 
(n+1)^{\card{\set{X}}}$. The proof of the direct part is completed by noting
that if $R>R_\rho(P,D)$, then for sufficiently small $\delta>0$
the RHS of~\eqref{eq:lossy_types} tends to one as $n$ tends to infinity.

To prove the converse, we fix for each $n\in\naturals$ an encoder/decoder pair $(f_n,\varphi_n)$ as
in~\eqref{eq:lossy_encoder} and~\eqref{eq:lossy_decoder}
satisfying~\eqref{eq:distortion}.
We may assume that 
\begin{equation}
\label{eq:disjoint}
\varphi_n(m)\cap \varphi_n(m') = \emptyset\quad\text{whenever $m\neq m'$}. 
\end{equation}
Indeed, if $m\neq m'$ and $\hat{x}^n \in \varphi_n(m)\cap \varphi_n(m')$,
then we can delete $\hat{x}^n$ from the larger of the two subsets, 
say~$\varphi_n(m)$, and map to $m'$ all the source sequences $x^n$ that were mapped
to $m$ by $f_n$ and satisfy $d(x^n,\hat{x}^n)\leq D$. This could only reduce the 
$\rho$-th moment of the number of performed tasks while preserving the
property~\eqref{eq:distortion}. 

Define the set
\begin{equation}
\set{Z}_n = \bigcup_{m=1}^{2^{nR}} \varphi_n(m).\label{eq:1042_5}
\end{equation}
The assumption~\eqref{eq:disjoint} implies that the union on the RHS of~\eqref{eq:1042_5} is disjoint. 
Consequently, we may define~$\mu_n(\hat{x}^n)$ for every $\hat{x}^n \in \set{Z}_n$ as the unique 
element of $\{1,\ldots,2^{nR}\}$ for which $\hat{x}^n
\in \varphi_n(\mu_n(\hat{x}^n))$. Moreover, \eqref{eq:distortion} guarantees the
existence of a mapping $g_n\colon \set{X}^n\to \set{Z}_n$ (not necessarily
unique) such that, for all $x^n\in\set{X}^n$,
\begin{equation}
\label{eq:984}
g_n(x^n) \in \varphi_n\bigl(f_n(x^n)\bigr)\quad\text{and}\quad d\bigl(x^n,g_n(x^n)\bigr)\leq D.
\end{equation}
We also define the PMF on $\set{Z}_n$, 
\begin{equation}
\tilde{P}_n(\hat{x}^n) = P^n\bigl(g_n^{-1}(\hat{x}^n)\bigr),\quad
\hat{x}^n\in\set{Z}_n,
\end{equation}
where 
\begin{equation}
g_n^{-1}(\hat{x}^n) = \{x^n\in \set{X}^n: g_n(x^n) = \hat{x}^n\}.
\end{equation}
With these definitions of $\mu_n$, $g_n$, and $\tilde{P}_n$, we have
\begin{align}
&\sum_{x^n\in\set{X}^n} P^n(x^n)
\bigcard{\varphi_n\bigl(f_n(x^n)\bigr)}^\rho\notag\\
&\quad=\sum_{\hat{x}^n\in\set{Z}_n}
P^n\bigl(g_n^{-1}(\hat{x}^n)\bigr)
\bigcard{\varphi_n\bigl(\mu_n(\hat{x}^n)\bigr)}^\rho\label{eq:1115_lossy_converse}\\
&\quad=\sum_{\hat{x}^n\in\set{Z}_n} \tilde{P}_n(\hat{x}^n)
\bigcard{\varphi_n\bigl(\mu_n(\hat{x}^n)\bigr)}^\rho\\
&\quad\geq 2^{\rho(H_{\rhotrans}(\tilde{P}_n)-nR)},\label{eq:978}
\end{align}
where the inequality~\eqref{eq:978} follows from~\eqref{eq:oneshot_lb} 
(with~$\set{Z}_n$, $\tilde{P}_n$, and~$\mu_n$ taking the roles of~$\set{X}$, $P$
and~$f$) by noting that~$\varphi_n=\mu_n^{-1}$.
In view of~\eqref{eq:1115_lossy_converse}--\eqref{eq:978} the converse is proved once we show that
\begin{equation}
\label{eq:renyi_of_bla}
H_{\rhotrans}(\tilde{P}_n) \geq n R_\rho(P,D). 
\end{equation}
To prove~\eqref{eq:renyi_of_bla}, note that on account
of~\eqref{eq:renyi_variation} we have for
every PMF $Q$ on $\set{Z}_n$
\begin{equation}
\label{eq:lossy_converse_13}
H_{\rhotrans}(\tilde{P}_n) \geq H(Q) - \rho^{-1} D(Q||\tilde{P}_n).
\end{equation}
The PMF $\tilde{P}_n$ can be written as
\begin{equation}
\label{eq:1023}
\tilde{P}_n=P^n W_n, 
\end{equation}
where $W_n$ is the deterministic channel from $\set{X}^n$ to
$\hat{\set{X}}^n$ induced by $g_n$:
\begin{equation}
W_n(\hat{x}^n|x^n) = \begin{cases} 1& \text{if $\hat{x}^n =  g_n(x^n)$},\\0&\text{otherwise.}\end{cases}
\end{equation}
Let $Q_\star$ be a PMF on $\set{X}$ that achieves the maximum in the definition
of $R_\rho(P,D)$, i.e., 
\begin{equation}
\label{eq:1061}
R_\rho(P,D) = R(Q_\star,D) - \rho^{-1} D( Q_\star||P).
\end{equation}
Substituting $Q_\star^n W_n$ for $Q$ in~\eqref{eq:lossy_converse_13} and
using~\eqref{eq:1023},
\begin{align}
H_{\rhotrans}(\tilde{P}_n) &\geq H(Q_\star^n W_n) - \rho^{-1}
D(Q_\star^n W_n||P^n W_n)\label{eq:lossy_converse_1354}\\
&\geq H(Q_\star^n W_n) - \rho^{-1}
D(Q_\star^n ||P^n)\label{eq:lossy_converse_1150}\\
&=H(Q_\star^n W_n) - n\rho^{-1} D(Q_\star
||P),\label{eq:lossy_converse_14}
\end{align}
where~\eqref{eq:lossy_converse_1150} follows from the Data Processing
Inequality \cite[Lemma 3.11]{csiszar2011information}. 
Let the source $\{\tilde{X}_i\}_{i=1}^\infty$ be IID $Q_\star$ 
and set $\hat{X}^n = g_n(\tilde{X}^n)$.
Then
\begin{align}
H(Q_\star^n W_n)&=H(\hat{X}^n)\label{eq:lossy_converse_1470}\\
&=I(\tilde{X}^n\wedge\hat{X}^n).\label{eq:lossy_converse_15}
\end{align}
By~\eqref{eq:984}, we have
\begin{equation}
\label{eq:lossy_converse_1175}
\E{d(\tilde{X}^n,\hat{X}^n)} \leq D,
\end{equation}
so applying~\cite[Theorem 9.2.1]{gallager1968information} (which is the main ingredient
in the classical rate-distortion converse) to the pair
$(\tilde{X}^n,\hat{X}^n)$ gives
\begin{align}
I(\tilde{X}^n\wedge\hat{X}^n) &\geq
n R\bigl(Q_\star,\E{d(\tilde{X}^n,\hat{X}^n)}\bigr)\label{eq:lossy_converse_1570}\\
&\geq n R(Q_\star,D),\label{eq:lossy_converse_16}
\end{align}
where~\eqref{eq:lossy_converse_16} follows from~\eqref{eq:lossy_converse_1175}
by the monotonicity of the rate-distortion function. 
Combining~\eqref{eq:lossy_converse_1570}--\eqref{eq:lossy_converse_16},
\eqref{eq:lossy_converse_1470}--\eqref{eq:lossy_converse_15},
\eqref{eq:lossy_converse_1354}--\eqref{eq:lossy_converse_14}, and~\eqref{eq:1061}
establishes~\eqref{eq:renyi_of_bla}.\qed

\section{Tasks with Costs}
\label{sec:costs}
We have so far assumed that every task 
requires an equal amount of effort. 
In this section, we discuss an extension where a
nonnegative, finite cost $c(x)$ is associated with each task $x\in\set{X}$.
For the sake of simplicity, we limit ourselves to IID sources and $\rho=1$. 

For an $n$-tuple of tasks $x^n \in \set{X}^n$, 
we denote by $c(x^n)$ the average cost per task:
\begin{equation}
c(x^n) = \frac{1}{n} \sum_{i=1}^n c(x_i).
\end{equation}
We still assume that $n$-tuples of tasks are describe using~$nR$ bits by an
encoder of the form $f\colon\set{X}^n\to\{1,\ldots,2^{nR}\}$, and that if
$x^n$ is assigned,
then all $n$-tuples in the set~$f^{-1}(f(x^n))$ are performed. 
Thus, if $x^n$ is assigned, then the average cost per assigned task is
\begin{equation}
c(f,x^n) \eqdef \sum_{\tilde{x}^n \in f^{-1}(f(x^n))} c(\tilde{x}^n).
\end{equation}
The following result extends Theorem~\ref{thm:main} to this setting (for IID
tasks and $\rho=1$). We focus on the case $\E{c(X_1)}>0$ because otherwise
we can achieve
\begin{equation}
\bigE{c(f,X^n)} =0
\end{equation}
using only one bit by setting $f(x^n)=1$ if $c(x^n)=0$ and $f(x^n)=2$ otherwise. 
\begin{theorem}
\label{thm:cost}
Let $\{X_i\}_{i=1}^\infty$ be IID with finite alphabet $\set{X}$ and
$\E{c(X_1)}>0$.
\begin{enumerate}
\item If $R> H_{1/2}(X_1)$, then there exist encoders $f_n \colon
\set{X}^n \to \{1,\ldots,2^{nR}\}$ such that
\begin{equation}
\lim_{n\to\infty}\bigE{c(f_n,X^n)} \to \E{c(X_1)}.
\end{equation}
\item If $R<H_{1/2}(X_1)$, then for any choice of encoders $f_n
\colon\set{X}^n \to \{1,\ldots,2^{nR}\}$,
\begin{equation}
\label{eq:cost_converse_blowup}
\lim_{n\to\infty}\bigE{c(f_n,X^n)} \to \infty.
\end{equation}
\end{enumerate}
\end{theorem}
\begin{proof}[Proof of Theorem~\ref{thm:cost}]
We begin with the case $R>H_{1/2}(X_1)$, i.e., the direct part. 
Let us denote by $c_{\textnormal{max}}$ the largest cost 
of any single task in $\set{X}$
\begin{equation}
c_{\textnormal{max}} = \max_{x\in\set{X}} c(x).
\end{equation}
Select a sequence $f_n\colon\set{X}^n
\to\{1,\ldots,2^{nR}\}$ as in the direct part of Theorem~\ref{thm:main} and
observe that 
\begin{align}
&\bigE{c(f_n,X^n)}\notag\\
&= \sum_{x^n \in \set{X}^n} P^n(x^n) c(f_n,x^n)\\
&= \sum_{x^n\in\set{X}^n} P^n(x^n) \Bigl( c(x^n) + \sum_{\tilde{x}^n 
\in f_n^{-1}(f_n(x^n)) \setminus\{x^n\}} c(\tilde{x}^n)\Bigr)\\
&=\bigE{c(X_1)} + \sum_{x^n\in\set{X}^n} P^n(x^n)\sum_{\tilde{x}^n 
\in f_n^{-1}(f_n(x^n)) \setminus\{x^n\}} c(\tilde{x}^n)\\
&\leq \bigE{c(X_1)} + c_{\textnormal{max}} \sum_{x^n\in\set{X}^n} P^n(x^n)
\card{f_n^{-1}(f_n(x^n)) \setminus\{x^n\}}\label{eq:cost_converse_1}\\
&=\bigE{c(X_1)} + c_{\textnormal{max}} 
\bigl(\bigE{\card{f_n^{-1}(f_n(X^n))}} - 1\bigr),\label{eq:cost_converse_1826}
\end{align}
and the second term on the RHS of~\eqref{eq:cost_converse_1826} 
tends to zero as $n\to\infty$ by
Theorem~\ref{thm:main}. 

We now turn to the case $R< H_{1/2}(X_1)$, i.e., the converse part. 
If the minimum cost of any single task $c_{\textnormal{min}}$ is positive,
then~\eqref{eq:cost_converse_blowup} follows from 
the converse part of Theorem~\ref{thm:main} by replacing
in~\eqref{eq:cost_converse_1} $c_{\textnormal{max}}$ 
with $c_{\textnormal{min}}$ and ``$\leq$'' with~``$\geq$''.
If at least one task has zero cost (i.e., $c_{\textnormal{min}}=0$),
then we need a different proof. 

The assumption $\E{c(X_1)}>0$ implies 
that there is some $x^\star\in\set{X}$ with
$P(x^\star)c(x^\star)>0$. 
Using H\"older's Inequality as in~\eqref{eq:hoelder_alt} with $p=q=2$,
$a(x)=\sqrt{P^n(x^n) c(f_n,x^n)} $, and $b(x) = \sqrt{c(x^n)/c(f_n,x^n)}$ gives
\begin{align}
&\sum_{x^n \in \set{X}^n} P^n(x^n) c(f_n,x^n)\notag\\
&\quad\geq \sum_{x^n:c(x^n)>0} c(x^n) P^n(x^n) \frac{c(f_n,x^n)}{c(x^n)}\\
&\quad\geq \frac{\bigl(\sum_{x^n:c(x^n)>0} \sqrt{c(x^n)P^n(x^n)}\bigr)^2}
{\sum_{x^n:c(x^n)>0} \frac{c(x^n)}{c(f_n,x^n)}}.\label{eq:cost_converse_1848}
\end{align}
To bound the denominator on the RHS of~\eqref{eq:cost_converse_1848}, 
observe that
\begin{align}
&\sum_{x^n:c(x^n)>0} \frac{c(x^n)}{c(f_n,x^n)}\notag\\
&= \sum_{m=1}^{2^{nR}} \sum_{x^n \in f_n^{-1}(m), c(x^n)>0} \frac{c(x^n)}{
\sum_{\tilde{x}^n \in f_n^{-1}(m)} c(\tilde{x}^n)}\\
&\leq 2^{nR},\label{eq:cost_converse_18548}
\end{align}
where the inequality follows because for some $m$ the set $\{x^n \in
f_n^{-1}(m): c(x^n)>0\}$ may be empty. 
Combining~\eqref{eq:cost_converse_18548} and~\eqref{eq:cost_converse_1848} gives
\begin{multline}
\sum_{x^n \in \set{X}^n} P^n(x^n) c(f_n,x^n)\\
\geq 2^{-nR} \biggl( \sum_{x^n:
c(x^n)>0} \sqrt{c(x^n)P^n(x^n)}\biggr)^2.\label{eq:cost_converse_1854}
\end{multline}
We can bound the sum on the RHS of~\eqref{eq:cost_converse_1854} as follows.
\begin{align}
&\sum_{x^n: c(x^n)>0} \sqrt{c(x^n)P^n(x^n)}\notag\\
& \geq \sqrt{\frac{c(x^\star)}{n}} \sum_{Q \in \set{P}_n(\set{X}), Q(x^\star)>0} 
\sum_{x^n \in T_Q}\sqrt{P^n(x^n)}\label{eq:cost_conv_cost_lb}\\
& \geq \sqrt{\frac{c(x^\star)}{n}} 
\max_{\substack{Q\in\set{P}_n(\set{X})\\Q(x^\star)>0}} 2^{n(H(Q)-\delta_n)}
2^{-\frac{n}{2}(D(Q||P)+H(Q))}\label{eq:cost_conv_id}\\
&=2^{\frac{n}{2}(\max_{Q\in\set{P}_n(\set{X}),Q(x^\star)>0}H(Q)-D(Q||P)-\delta_n')}\\
&=2^{\frac{n}{2}(H_{1/2}(X_1) - \epsilon_n - \delta'_n)},\label{eq:cost_converse_1867}
\end{align}
where $\delta_n = n^{-1}\card{\set{X}} \log(n+1)$, where $\delta_n'
=2\delta_n+n^{-1}\log(n/c(x^\star))$, and
where $\epsilon_n\to 0$ as $n\to\infty$. 
Here,~\eqref{eq:cost_conv_cost_lb} follows because if 
$x^n \in T_Q$ and $Q(x^\star)>0$, then
$x_i=x^\star$ for at least one~$i$ and hence $c(x^n)\geq c(x^\star)/n>0$;
\eqref{eq:cost_conv_id} follows because $P^n(x^n) = 2^{-n(D(Q||P)+H(Q))}$ when
$x^n \in T_Q$, and because $\card{T_Q} \geq 2^{n(H(Q)-\delta_n)}$;
\eqref{eq:cost_converse_1867} follows from~\eqref{eq:renyi_variation} 
because the set of rational PMFs $Q$ with $Q(x^\star)>0$ 
is dense in the set of all PMFs on~$\set{X}$, and 
$H(Q)-D(Q||P)$ is continuous in $Q$ (provided that $Q(x)=0$ whenever $P(x)=0$,
which is certainly satisfied by the maximizing $Q$ in~\eqref{eq:renyi_variation}).
Combining~\eqref{eq:cost_converse_1867} and~\eqref{eq:cost_converse_1854}
completes the proof of the converse. 
\end{proof}

\appendices
\section{Proof of Proposition~\ref{prop:sufficiency}}
\label{sec:proposition}
Since the labels do not matter, we may assume for convenience of notation that
$\set{X}=\{1,\ldots,\card{\set{X}}\}$ and 
\begin{equation}
\label{eq:increasing}
\lambda(1) \leq \lambda(2) \leq \cdots \leq \lambda(\card{\set{X}}).
\end{equation}
We construct a partition of $\set{X}$ as follows. 
The first subset is
\begin{equation}
\label{eq:L0}
\set{L}_0=\{x\in\set{X}:\lambda(x) \geq \card{\set{X}}\}. 
\end{equation}
If $\set{X}=\set{L}_0$, then the construction is complete and~\eqref{eq:alpha} and~\eqref{eq:card_bound} are 
clearly satisfied. Otherwise we follow the steps below to construct additional
subsets
$\set{L}_1,\ldots,\set{L}_M$. (Note that if $\set{L}_0 \neq \set{X}$, then
$\set{X} \setminus \set{L}_0 = \{1,\ldots, \card{\set{X}}-\card{\set{L}_0}\}$.)
\begin{quote}
\emph{Step $1$}: If 
\begin{equation}
\card{\set{X} \setminus \set{L}_0} \leq
\lambda(1), 
\end{equation}
then we complete the construction by setting 
$\set{L}_{1} = \set{X}\setminus\set{L}_0$ and $M=1$. Otherwise 
we set
\begin{equation}
\set{L}_1 = \bigl\{1,\ldots,\lambda(1)\bigr\}
\end{equation}
and go to Step $2$.\\
\emph{Step $m \geq 2$}: If 
\begin{equation}
\biggl|\set{X} \setminus \bigcup_{i=0}^{m-1} \set{L}_i\biggr| \leq
\lambda(\card{\set{L}_1}+\ldots+\card{\set{L}_{m-1}}+1), 
\end{equation}
then we complete the construction by setting $\set{L}_{m} =
\set{X}\setminus\bigcup_{i=0}^{m-1} \set{L}_i$ and $M=m$. Otherwise 
we let~$\set{L}_{m}$ contain the
$\lambda(\card{\set{L}_1}+\ldots+\card{\set{L}_{m-1}}+1)$ smallest elements of
$\set{X}\setminus\bigcup_{i=0}^{m-1} \set{L}_i$, i.e., we set
\begin{multline}
\set{L}_m = \bigl\{\card{\set{L}_1}+\ldots+\card{\set{L}_{m-1}}+1,\ldots,\\
\card{\set{L}_1}+\ldots+\card{\set{L}_{m-1}}+\lambda(\card{\set{L}_1}+\ldots+\card{\set{L}_{m-1}}+1)\bigr\}
\end{multline}
and go to Step $m+1$.  
\end{quote}
We next verify that~\eqref{eq:card_bound} is satisfied and that
the total number of subsets $M+1$ does not exceed~\eqref{eq:alpha}. 
Clearly,~$L(x) \leq \card{\set{X}}$ for every $x\in \set{X}$, so to
prove~\eqref{eq:card_bound} we check
that $L(x) \leq \lambda(x)$ for every $x\in\set{X}$. 
From~\eqref{eq:L0} it is clear that $L(x) \leq
\lambda(x)$ for all $x\in \set{L}_0$. Let $k(x)$ denote the smallest element in
the subset containing $x$. Then $L(x) \leq \lambda(k(x))$ for all $x\in
\bigcup_{m=1}^M \set{L}_m$ by construction (the inequality can be strict only if $x\in \set{L}_M$), and
since~$k(x) \leq x$, we have $\lambda(k(x)) \leq \lambda(x)$ by the
assumption~\eqref{eq:increasing}, and hence $L(x) \leq \lambda(x)$ for all $x\in
\set{X}$. 

It remains to
check that $M+1$ does not exceed~\eqref{eq:alpha}. This is clearly true when
$M=1$, so we assume that~$M \geq 2$. 
Fix an arbitrary $\alpha>1$ and let $\set{M}$ be the set of indices $m \in \{1,\ldots,M-1\}$ such that 
there is an~$x\in \set{L}_m$ with~$\lambda(x) > \alpha \lambda(k(x))$. We next show that
\begin{equation}
\label{eq:bound_on_M_1244}
\card{\set{M}} < \log_\alpha \card{\set{X}}. 
\end{equation}
To this end, enumerate the indices in $\set{M}$ as $m_1<m_2 < \cdots < m_{\card{\set{M}}}$.
For each $i\in \{1,\ldots,\card{\set{M}}\}$ select some~$x_{i}\in \set{L}_{m_i}$
for which $\lambda(x_{i}) > \alpha \lambda(k(x_i))$. Then 
\begin{align}
\lambda(x_{1}) &> \alpha\lambda(k(x_1))\\
&\geq \alpha.
\end{align}
Note that if $m<m'$ and $x\in \set{L}_m$ and $x'\in\set{L}_{m'}$, then $x<x'$.
Thus, $x_1<k(x_2)$ because $x_1 \in \set{L}_{m_1}$, $k(x_2) \in \set{L}_{m_2}$,
and $m_1<m_2$. Consequently, $\lambda(x_1) \leq \lambda(k(x_2))$ and hence
\begin{align}
\label{eq:632_24}
\lambda(x_{2}) &> \alpha \lambda(k(x_{2}))\\
&\geq \alpha \lambda(x_{1})\\
&> \alpha^2.
\end{align}
Iterating this argument shows that
\begin{align}
\label{eq:lambda_alpha_1238}
\lambda(x_{\card{\set{M}}})> \alpha^{\card{\set{M}}}.
\end{align}
And since $\lambda(x) < \card{\set{X}}$ for every $x\notin \set{L}_0$ by~\eqref{eq:L0}, 
the desired inequality~\eqref{eq:bound_on_M_1244} follows from~\eqref{eq:lambda_alpha_1238}.
Let $\comp{\set{M}}$ denote the complement of $\set{M}$ in $\{1,\ldots, M-1\}$.
Using Proposition~\ref{prop:count_lists} and the fact that $L(x) =
\lambda(k(x))\geq \lambda(x)/\alpha$ for all $x\in \bigcup_{m \in
\setcomp{\set{M}}} \set{L}_m$, 
\begin{align}
M &=  \sum_{x\in\bigcup_{m=1}^M \set{L}_m}
\frac{1}{L(x)}\label{eq:count_sets_proof_1307}\\
&= 1+\card{\set{M}}+ \sum_{x\in \bigcup_{m\in\setcomp{\set{M}}}\set{L}_m} \frac{1}{L(x)}\\
&\leq 1+\card{\set{M}}+ \alpha\sum_{x\in \bigcup_{m\in\setcomp{\set{M}}}\set{L}_m}
\frac{1}{\lambda(x)}\\
&< 1 + \log_\alpha \card{\set{X}} + \alpha \mu,\label{eq:723_final}
\end{align}
where~\eqref{eq:723_final} follows from~\eqref{eq:bound_on_M_1244}
and the hypothesis of the proposition~\eqref{eq:sum_of_recip_equals_mu}.
Since $M+1$ is an integer and $\alpha>1$ is arbitrary, it follows
from~\eqref{eq:count_sets_proof_1307}--\eqref{eq:723_final} that $M+1$ is upper-bounded
by~\eqref{eq:alpha}.\qed


\section{Proof of~\eqref{eq:cond_renyi_identity}}
\label{sec:cond_renyi_identity}
We first show that $H(V|Q) - \rho^{-1} D(Q\circ V||P_{X,Y}) \leq
H_{\rhotrans}(X_1|Y_1)$ for 
every $Q \in \set{P}(\set{Y})$ and $V \in \set{P}(\set{X}|\set{Y})$.  This is
clearly true when $D(Q\circ V||P_{X,Y})=\infty$, so we may assume that
$P_{X,Y}(x,y)=0$ implies $Q(y)V(x|y)=0$, and hence that $P_Y(y)=0$ implies
$Q(y)=0$. Now observe that
\begin{align}
&H(V|Q) - \rho^{-1} D(Q\circ V||P_{X,Y})\notag\\
&=\frac{1+\rho}{\rho} \sum_{y\in\set{Y}}
Q(y) \sum_{x\in\set{X}} V(x|y) \log
\frac{P_{X|Y}(x|y)^{\frac{1}{1+\rho}}}{V(x|y)}\notag\\
&\qquad- \frac{1}{\rho} \sum_{y\in\set{Y}}
Q(y) \log \frac{Q(y)}{P_Y(y)}\\
&\leq \frac{1+\rho}{\rho} \sum_{y\in\set{Y}} Q(y) \log \sum_{x\in\set{X}}
P_{X|Y}(x|y)^{\frac{1}{1+\rho}}\notag\\
&\qquad- \frac{1}{\rho}\sum_{y\in\set{Y}} Q(y) \log
\frac{Q(y)}{P_Y(y)}\label{eq:cond_renyi_identity_proof_1}\\
&= \frac{1}{\rho} \sum_{y\in\set{Y}} Q(y) \log \frac{P_Y(y) \bigl(\sum_{x\in\set{X}}
P_{X|Y}(x|y)^{\frac{1}{1+\rho}}\bigr)^{1+\rho}}{Q(y)}\\
&\leq \frac{1}{\rho} \log \sum_{y\in \set{Y}} P_Y(y) \biggl(\sum_{x\in\set{X}}
P_{X|Y}(x|y)^{\frac{1}{1+\rho}}\biggr)^{1+\rho}\label{eq:cond_renyi_identity_proof_2}\\
&= H_{\rhotrans}(X_1|Y_1),
\end{align}
where~\eqref{eq:cond_renyi_identity_proof_1}
and~\eqref{eq:cond_renyi_identity_proof_2} follow from Jensen's Inequality. 
The proof is completed by noting that equality is attained in both inequalities 
by the choice
\begin{equation}
Q(y) = \frac{P_Y(y) \bigl(\sum_{x\in\set{X}} P_{X|Y}(x|y)^{\frac{1}{1+\rho}}
\bigr)^{1+\rho}}{\sum_{y' \in \set{Y}} P_Y(y') \bigl(\sum_{x\in\set{X}}
P_{X|Y}(x|y')^{\frac{1}{1+\rho}}\bigr)^{1+\rho}},
\end{equation}
and
\begin{equation}
\label{eq:cond_renyi_identity_proof_3}
V(x|y) = \frac{P_{X|Y}(x|y)^{\frac{1}{1+\rho}}}{\sum_{x'\in\set{X}}
P_{X|Y}(x'|y)^{\frac{1}{1+\rho}}}, \quad Q(y)>0.
\end{equation}
(Note that $P(y)>0$ when $Q(y)>0$ so the RHS
of~\eqref{eq:cond_renyi_identity_proof_3} makes sense. How we define $V(x|y)$ when $Q(y)=0$ does not matter.) \qed

%


\ifCLASSOPTIONcaptionsoff
  \newpage
\fi



\bibliographystyle{IEEEtran}
\bibliography{IEEEabrv,tasks_jrnl}
\end{document}